\documentclass[a4paper, 11pt]{article}

\usepackage[latin1]{inputenc}
\usepackage{amstext, amssymb, amsthm, amsfonts, latexsym,bbm}
\usepackage[tbtags]{amsmath}
\usepackage{hyperref}
\usepackage{color}
\usepackage{graphicx}
\usepackage{enumerate}
\usepackage[backgroundcolor=yellow]{todonotes}
\usepackage{caption, subcaption}
\usepackage{authblk}
\usepackage[normalem]{ulem}

\def\E{\mathbb{E}}
\def\R{\mathbb{R}}

\def\S{\mathbb{S}}
\def\C{\mathbb{C}}
\def\B{\mathbb{B}}

\def\RP{{\bf P^{1}}\mathbb{R}}
\def\Prob{\mathbb{P}}

\newtheorem{theorem}{Theorem}
\newtheorem{proposition}[theorem]{Proposition}
\newtheorem{lemma}[theorem]{Lemma}
\newtheorem{corollary}[theorem]{Corollary}

\theoremstyle{definition}
\newtheorem{rem}[theorem]{Remark}
\newtheorem{definition}[theorem]{Definition}

\begin{document}

\title{Mean field repulsive Kuramoto models:\\
Phase locking and   spatial signs}

\author[1]{Corina Ciobotaru, Linard Hoessly, Christian Mazza and Xavier Richard }
\affil[1]{University of Fribourg, Department of Mathematics, Chemin du Mus\'ee 23, CH-1700 Fribourg, Switzerland}
\date{\today}

\maketitle

\begin{abstract}
The phenomenon of self-synchronization in  populations of oscillatory units appears naturally in neurosciences. However, in some situations, the formation of a coherent state is damaging. In this article we study a repulsive mean-field Kuramoto model that describes the time evolution of $n$ points on the unit circle, which are transformed into incoherent phase-locked states. It has been recently shown that such systems can be reduced to a three-dimensional system of ordinary differential equations, whose mathematical structure is strongly related to hyperbolic geometry. The orbits of the Kuramoto dynamical system are then described by a flow of M\"obius transformations. We show this underlying dynamic performs statistical inference by computing dynamically M-estimates of scatter matrices.
 We also describe the limiting phase-locked states for random initial conditions using Tyler's transformation matrix. Moreover, we show the repulsive Kuramoto model performs dynamically not only robust covariance matrix estimation, but also data processing: 
the initial configuration of the $n$ points is transformed by the dynamic into a
 limiting phase-locked state that surprisingly equals the spatial signs  from nonparametric statistics. That makes the sign empirical covariance matrix to equal $\frac{1}{2}{\rm id}_2$, the variance-covariance matrix of a random vector that is uniformly distributed on the unit circle.
\end{abstract}

\section{Introduction}
The phenomenon of self-synchronization in  populations of oscillatory units appears naturally in neurosciences. Nowadays,  large systems of interacting oscillatory elements like systems describing pedestrian synchrony on foot bridges or neuronal and brain rhythms are becoming increasingly popular both in empirical and theoretical research. 

Winfree \cite{Winfree1967} proposed a mathematical model to describe collective synchrony of large populations of biological oscillators. Later, Kuramoto \cite{Kuramoto1975} proposed a model, which can be seen as a weak-coupling limit of Winfree model (see, e.g.  \cite{Pikovsky2015}). Such models have the generic form
\begin{equation}\label{KuramotoGeneral}
\frac{{\rm d}\theta_j(t)}{{\rm d}t}=
f+g \cos(\theta_j(t))+h\sin(\theta_j(t)), \ j=1,\cdots, n,
\end{equation}
where $f, g$ and $h$ are smooth functions of the angles $(\theta_1(t),\cdots, \theta_n(t))$.
In what follows we consider the following simplified version of the Kuramoto model (see, e.g., \cite{Watanabe1994,Strogatz2000,Pikovsky2015,Chen2017})
\begin{equation}\label{Kuramoto1}
\frac{{\rm d}\theta_j(t)}{{\rm d}t}=\omega + \frac{1}{n}\sum_{k=1}^n \cos(\theta_k(t)-\theta_j(t)-\delta),\ j=1,\cdots, n,
\end{equation}
where $\omega$ is a constant and $\delta$ is a constant angle.

 Many interesting features emerge when performing simulations. Such systems are known to have potential for synchronizing the oscillators in the attractive case $\delta =\pi/2$, see e.g., \cite{Strogatz2000,Acebron2005,Rodrigues2016}. For equation (\ref{Kuramoto1}) we will mainly focus on the case when $\delta = -\pi/2$; this is called the \textit{repulsive case}. Repulsive Kuramoto dynamics lead to phase incoherence with phase-locked limiting states,
a property which plays an important role in natural systems, see, e.g., \cite{Tsimring2005,Hong2011,Hong2011B,Louzada2012,Giron2016,DiYuan2017}.

It turns out \cite{Watanabe1994,Strogatz2000} that the  dynamics of equation (\ref{Kuramoto1})   can be reduced to a 3-dimensional system of equations using M\"obius transformations. The idea of these authors consists in using the so-called Watanabe--Strogatz (WS) transformation that is the main ingredient for reducing the system. Later, the authors of \cite{Chen2017} have shown that the dynamics of systems of the form given in (\ref{Kuramoto1}) can be seen as gradient dynamics where the gradient is taken with respect to the hyperbolic metric
on the Poincar\'e disc. For the rest of the article we let $D:= \{z \in \mathbb{C} \; \vert \; \vert z\vert < 1\}$ and denote its closure by $\overline{D}$.

The main idea behind the (WS) transformation is to find M\"obius transformations $M(t)$ and base points $\beta_j\in\S^1$, $j=1,\cdots,n$, which are related to $M(0)$ and to the initial conditions $\theta_j(0)$, $j=1,\cdots,n$, so that  $z_j(t)=e^{i\theta_j(t)}$ (where $\theta_j(t)$ solves (\ref{Kuramoto1})) can be written as
$$z_j(t)=M(t)(\beta_j),\ j=1,\cdots,n.$$
The authors of \cite{Chen2017} focus on transformations of the generic form
\begin{equation}\label{FirsMob}
M(t)(z)=\eta(t)\frac{z-w(t)}{1-\bar w(t)z},
\end{equation}
where $\eta(t)\in\S^1$ and $w(t)\in D$, for all  $t$.
Plugging this Ansatz in (\ref{Kuramoto1}) leads to two ordinary differential equations (o.d.e.) that $w(t)$ and $\eta(t)$ must satisfy
(see Section \ref{section::conformal-barycenter}). One of these two equations takes the form
\begin{equation}\label{odeIntro}
\frac{{\rm d}w(t)}{{\rm d}t}=
-\frac{1}{2}(1-\vert w(t)\vert^2)e^{i(\frac{\pi}{2}-\delta})\bar\eta(t)\frac{1}{n}\sum_{j=1}^n
M(t)\beta_j.
\end{equation}
 In the repulsive case, these authors then noticed the above differential equation is of gradient type and the solution (if exists and is unique) to the gradient equals zero maximizes the function of $w\in D$ given by
$$L(w;p)=\sum_{j=1}^n \ln(\frac{1-\vert w\vert^2}{\vert \beta_j-w\vert^2}),
$$
where $p=(\beta_1,\cdots,\beta_n)\in (\S^1)^n$ is the base point. 
They also showed that the maximizer, when it exists and is unique, is the conformal barycenter (see \cite{Douady1986}) of the sample empirical measure $\mu_n=\frac{1}{n}\sum_{j=1}^n \delta_{\beta_j}$ which is defined on the unit circle $\S^1$. Then they deduce many new and interesting properties of the solution to (\ref{Kuramoto1}).\\

Starting from these observations, we show the above potential function $L(w;p)$ is the log-likelihood function of the sample 
$\{\beta_1,\cdots,\beta_n\} \subset \S^1$, for an underlying statistical model where the data points are independent and identically distributed (i.i.d.) of law given by a {\it circular Cauchy distribution} of parameter $w$, which is also known as a Poisson kernel in physics. In Section \ref{CircularCauchy}, we define precisely circular Cauchy distributions, the related likelihood function, and the notion of circular Cauchy M-functional  ${\rm MLE}(\cdot)$. 
We show the equilibrium point of the (o.d.e.) (\ref{odeIntro}) corresponds to the M-estimate ${\rm MLE}(\mu_n)$ for the empirical measure $\mu_n$ of the sample $\{\beta_1,\cdots,\beta_n\}$, so that, repulsive Kuramoto dynamics perform estimation of the parameter $w$. 

\medskip
It turns out that circular Cauchy distributions are strongly related to {\it central angular Gaussian} laws which are defined on the projective line $\RP$. Section \ref{Section::centralAngularGaussian} recalls the main definitions and results related to such distributions. The associated random variables are obtained from bivariate centred normal distributions through projection on $\RP$.  We define the notions of likelihood functions, M-functionals and M-estimates $\Sigma(\Prob)$ associated with probability measures $\Prob$ on $\RP$.  The natural parameter of central angular Gaussian distributions is the covariance matrix $\Sigma$ of the underlying bivariate normal law.
We then relate in Theorem \ref{Gauss_Cauchy_max_like} the M-estimate ${\rm MLE}(\mu_n)$ and the M-estimate  $\Sigma(\Prob_n)$ , where 
$\Prob_n$ is the empirical measure of a sample $\{[y_1],\cdots,[y_n]\}$ of elements $y_j\in\S^1$ with $y_j^2 =\beta_j$. Hence, Kuramoto dynamics perform statistical inference by providing dynamically statistically robust estimates of covariance matrices.\\

 Section \ref{Section::AsymptoticMax} provides relevant asymptotic results in the large $t$ and $n$ limits. For example, a complete description of phase-locked solutions to (\ref{Kuramoto1}) is given, where $z_j(t)$ is asymptotically phase-locked in states of the form $M_{{\rm MLE}(\mu_n)}(\beta_j)$. We also consider the situation where the base point $p$ has i.i.d. components $\beta_j$ of law $\mu$, and describe the almost sure convergence of ${\rm MLE}(\mu_n)$ toward 
${\rm MLE}(\mu)$ and establish in Proposition \ref{KuramotoRandom1} and Corollary \ref{Cor1}, using M-estimator asymptotic theory, central limit theorems for ${\rm MLE}(\mu_n)$. These results describe in a precise way phase-locking in  repulsive Kuramoto dynamics. \\

Finally, Section \ref{Section::SpatialSigns} introduces notions from nonparametric multivariate statistics which permit to relate phase-locked states to affine invariant statistics from nonparametric sign test theory. We show in this way in Theorem \ref{LinkB} that repulsive Kuramoto dynamics are not only performing dynamically robust covariance matrix estimation but also perform data processing: 
the initial configuration $(z_j(0))_{j=1,\cdots,n}\in (\S^1)^n$ is transformed by the dynamics into a
 limiting phase-locked state that surprisingly equals the square of
   spatial signs  $(S_j^2)_{j=1,\cdots,n}\in (\S^1)^n$ from nonparametric statistics,
 that makes the sign empirical covariance matrix to equal $\frac{1}{2}{\rm id}_2$, the variance-covariance matrix of a random vector that is uniformly distributed on the unit circle. This picture is then completed in Corollaries \ref{S2Uniform} and \ref{IncoherentSpatialSign} where new statistical results on spatial signs are provided.

\section{Conformal barycenters\label{section::conformal-barycenter}}

Recall first an element of the sub-group $G$ of M\"obius transformations preserving the closure $\overline{D}$ of the Poincar\'e disc $D$  is of the form
\begin{equation}\label{Mobius}
M(z):=\eta\frac{z-w}{1-\bar wz}, \; \; \;  \vert \eta\vert =1 \text{ and } \vert w \vert <1, \; z\in \overline{D},
\end{equation} and we denote 
\begin{equation}\label{Mobius1}
M_{w}(z):=\frac{z-w}{1-\bar wz}, \; \; \;   \vert w \vert <1, \; z\in \overline{D}.
\end{equation}

Moreover, it is well known that $G$ acts transitively on the set of all triples of pairwise disjoint points on the boundary $\S^1$ of $\overline{D}$.  

Let $z_j(t):=e^{i\theta_j(t)}$, where the $\theta_j(t)$ solve (\ref{Kuramoto1}). For a base point $p=(\beta_1,\cdots,\beta_n)\in (\S^1)^n$, 
the authors of \cite{Chen2017} used the  (WS)  reduction principle to obtain explicit formulas of the form
\begin{equation}\label{Mobius1}
z_j(t)=M(t)(\beta_j)=\eta(t)\frac{\beta_j-w(t)}{1-\bar w(t) \beta_j},\ j=1,\cdots,n,
\end{equation}
for $M(t) \in G$ a one-parameter family of  M\"obius transformations.
Notice the initial conditions $\theta_j(0)$ and the base point $p$ are related to each other according to the relation

\begin{equation}\label{RelationInit}
\beta_j = M_{-w(0)}(\bar\eta(0) e^{i\theta_j(0)}),\ j=1,\cdots,n.
\end{equation}
To find the one-parameter family $M(t) \in G$ the (o.d.e.) given by equation (\ref{Kuramoto1}) is then reduced to the  system 
\begin{equation}\label{odew}
\frac{{\rm d}w(t)}{{\rm d}t}=
-\frac{1}{2}(1-\vert w(t)\vert^2)e^{i(\frac{\pi}{2}-\delta)}\bar\eta(t)\frac{1}{n}\sum_{j=1}^n
M(t)\beta_j,
\end{equation}
\begin{equation}\label{odewrot}
\frac{{\rm d}\eta(t)}{{\rm d}t}=
i \omega \eta(t)-\frac{1}{2}\Big(\bar w(t) {\cal A}-w(t)\bar{\cal A}\eta(t)^2\Big),
\end{equation}
where ${\cal A}$ is evaluated at $M(t)p$.
\medskip

In the repulsive case $\delta =-\pi/2$, (\ref{odew}) becomes
\begin{equation}
\label{equ::reduced}
\frac{{\rm d}w(t)}{{\rm d}t}=
- \frac{1}{2}(1-\vert w(t)\vert^2) \frac{1}{n}\sum_{j=1}^n
M_{w(t)}(\beta_j).
\end{equation}
By definition an element $w^*(p) \in \overline{D}$  is called an \textit{equilibrium} for equation (\ref{equ::reduced}) if either $\vert w^*(p)\vert =1$ or
\begin{equation}\label{incoherent1}
\sum_{j=1}^n M_{w^*(p)}(\beta_j)=0.
\end{equation}
 Hence, when $w^*(p)\in D$ is an equilibrium for equation (\ref{equ::reduced}),  $z^*:=\Big(M_{w^*(p)}(\beta_j)\Big)_j$  
belongs to the manifold of {\it incoherent states}
\begin{equation}\label{IncoherentManifold}
z^*\in {\cal M}:=\{z=(z_1,\cdots,z_n)\in\mathbb{S}^n;\ \sum_{j=1}^n z_j=0\}.
\end{equation}

The authors of \cite{Chen2017} recall and develop further the link that exists between  the representation 
(\ref{Mobius1}) of the solution to the (o.d.e.) (\ref{Kuramoto1}) and the hyperbolic geometry on the unit disc  $D$. It is shown that relative to the hyperbolic metric on $D$ the reduced (o.d.e.) given in (\ref{odew}) is the gradient of the potential
\begin{equation}\label{Potential}
H(w):=-\frac{1}{n}\sum_{j=1}^n \ln(\frac{1-\vert w\vert^2}{\vert \beta_j-w\vert^2}).
\end{equation}

The equilibrium $w^*(p) \in \overline{D}$ for equation (\ref{equ::reduced}) (if it exists and is unique)  maximizes the function
\begin{equation}\label{Potential2}
L(w;p):=\sum_{j=1}^n \ln(\frac{1-\vert w\vert^2}{\vert \beta_j-w\vert^2})
\end{equation}
and satisfies  (\ref{incoherent1}). If there is a unique maximizer $w^*(p)$, the gradient form of the (o.d.e.) (\ref{equ::reduced}) implies that
$w(t)\longrightarrow w^*(p)$ as $t\to\infty$.

Then the authors of \cite{Chen2017} also show that $w^*(p)$ is related to the \textit{conformal barycenter} of the atomic probability measure 
$\mu = \frac{1}{n}\sum_{j=1}^{n} \delta_{\beta_j}$ on $\mathbb{S}^1$. Let us first recall the notion of conformal barycenter following \cite{Douady1986}:
Let $\mu$ be a probability measure on the unit circle $\mathbb{S}^1$. When $\mu$ is with no atoms, the authors of \cite{Douady1986} assign to $\mu$ an element $B(\mu)\in D$. It is defined by considering the vector field $\xi_\mu$ on $D$  such that
$$\xi_\mu(0):=\int_{\mathbb{S}^1}\xi {\rm d}\mu(\xi),$$
and, more generally, for any $w\in D$,
\begin{equation}\label{VectorField}
\xi_{\mu}(w) :=(1-\vert w\vert^2)\int_{\mathbb{S}^1}\frac{\xi-w}{1-\bar w \xi}{\rm d}\mu(\xi)
            =(1-\vert w\vert^2)\int_{\mathbb{S}^1}M_w(\xi){\rm d}\mu(\xi).
\end{equation}

The following definition considers the notion of conformal barycenter for probability measures $\mu$ with no atoms. We will extend this definition to any probability measure on $\mathbb{S}^1$ when dealing with M-functionals for circular Cauchy and central angular Gaussian distributions. 

By Proposition 1 of \cite{Douady1986}, if $\mu$ has no atoms, then $\xi_{\mu}(w)=0$ has a unique solution $B(\mu)$ in  $D$.
\begin{definition}[\cite{Douady1986}]\label{defi::conformal-barycenter}
Suppose that $\mu$ has no atoms. Then the unique zero $B(\mu) \in D$ of the vector field $\xi_\mu$  is called the \textit{conformal barycenter} of the probability measure $\mu$.
\end{definition}

For the repulsive case, we see that (\ref{VectorField}) is nothing but the vector field associated with the differential equation given in (\ref{equ::reduced}) for the empirical measure $\mu =\mu_n= \frac{1}{n} \sum_{j=1}^{n} \delta_{\beta_j}$. The authors of \cite{Chen2017} have proven that the repulsive model has a unique equilibrium point $w^*(p)$ when the data 
$p=(\beta_1,\cdots,\beta_n)$ does not have a majority cluster of at least $n/2$ equal $\beta_j$. Then this element $w^*(p)$  is the unique zero of the vector field $\xi_\mu$.

\section{The circular Cauchy distribution on $\S^1$\label{CircularCauchy}}

We will show that  $L(w;p)$
 is the empirical log-likelihood of the sample $\{\beta_1,\cdots,\beta_n\}$ under the circular Cauchy distribution of parameter $w$, and therefore, the equilibrium point $w^*(p)$, when it exists and is unique, is the sample M-estimate of the unknown parameter $w$.
 
 Using this new way of looking at Kuramoto orbits, we will obtain new asymptotic results as $n\to\infty$ from mathematical statistics. Conversely, we will obtain new results in nonparametric multivariate statistics using mathematical results from Kuramoto dynamical system theory. 

\medskip
To define the standard Cauchy and Circular Cauchy distributions we follow \cite{McCullagh1996}. For the univariate and multivariate central angular Cauchy distributions we use the notions given in \cite{Kent1988} and \cite{Auderset2005}.

\subsection{The univariate Cauchy law on $\R$}

A real random variable $Y$ is said to follow a Cauchy distribution of parameter $\lambda = \mu +i\sigma$, where $\mu\in\R$ and $\sigma >0$ when it has the density
$$f_\lambda(y):=\frac{\sigma}{\pi((y-\mu)^2+\sigma^2)} \; .$$
The standard notation is $Y\sim C(\lambda)$. 

\subsection{The Circular Cauchy distribution}

We define in what follows the notion of {\it circular Cauchy distribution} which is a distribution on $\S^1$ or on $[0,2\pi]$ if one uses polar coordinates. In most works in physics dealing with the Kuramoto model, this distribution is known as a {\it Poisson kernel}, see, e.g., \cite{Marvel2009B,Marvel2009}.
Let $Y\sim C(\lambda)$ and consider the new random variable
$$Z:=\frac{1+iY}{1-iY}\in\C,$$
which is such that $\vert Z\vert =1$. Then it can be shown that $Z$ has the density
\begin{equation}\label{Wrapped}
f_w(z):=\frac{1-\vert w\vert^2}{2\pi\vert z-w\vert^2},
\end{equation}
where
$w:=(1+i\lambda)/(1-i\lambda)$ and $z$ is on the unit circle.
This random variable $Z$ is said to have a \textit{circular (or wrapped) Cauchy distribution}.
 The standard notation is $Z\sim C^*(w)$. It turns out that Z has the same law as the random variable
$$W:=e^{iY},$$
which is known as the {\it wrapped Cauchy random variable}, see e.g., \cite{Kent1988}. Let $z=e^{i\theta}$ for $0\le \theta\le 2\pi$ and set $w =\rho e^{i\alpha}$, $0\le \alpha \le 2\pi$. Then the density of $Z$ becomes
\begin{equation}\label{CircularPolarDensity}
g_{w}(\theta):=\frac{1-\rho^2}{2\pi(1+\rho^2-2\rho\cos(\theta-\alpha))}.
\end{equation}

According to \cite{McCullagh1996}, the family of circular Cauchy distributions is closed under the action of the sub-group of M\"obius transformations which preserve the unit disc $\overline{D}$. Any such M\"obius transformations is the composition of the following two transformations
\begin{equation}\label{TransformationsUnitCircle}
Z \mapsto e^{i\alpha}Z \hbox{ and } Z\mapsto \frac{Z-\gamma}{\bar\gamma Z-1},\ \vert\gamma\vert < 1.
\end{equation}
 If the random variable $Z$ is circular Cauchy $C^*(\psi)$, then the induced distributions are
\begin{equation}\label{InducedDistributions}
C^*(e^{i\alpha}\psi) \hbox{ and }C^*(\frac{\psi-\gamma}{\bar\gamma\psi-1}).
\end{equation}
The basic family of M\"obius transformations which occur in the Kuramoto model are such that
\begin{equation}\label{MobCauchy}
M_w(z)=e^{i\pi}\frac{z-w}{\bar w z-1},
\end{equation}
so that, if $Z$ is circular Cauchy $C^*(\psi)$, then $M_w(Z)$ is  circular Cauchy
$C^*(M_w(\psi))$.

\begin{rem}
The authors of \cite{Marvel2009B,Marvel2009} have considered the large $n$ limit of models which are similar to the simpler model given by (\ref{Kuramoto1}). They have shown the existence of a two-dimensional invariant manifold which consists in so-called Poisson kernels (or circular Cauchy distributions here). One can check this is a consequence of (\ref{InducedDistributions}).
\end{rem} 

\subsection{Circular Cauchy M-functionals \label{MaxLikeliCircular}}

Let $\mu$ be a probability measure on $\mathbb{S}^1$, typically (but not necessarily) the empirical distribution of a sample. The {\it circular Cauchy log-likelihood function} for $\mu$ is the expectation of 
$\ln f_w$ evaluated at a random point $\xi$ of law $\mu$:
\begin{equation}\label{CircularCauchyLokLikelihood}
L(w;\mu): = \int_{\mathbb{S}^1} \ln (f_w(\xi)) {\rm d}\mu(\xi).
\end{equation}
The {\it circular Cauchy M-functional} 
is the functional that assigns the maximizer of
 $L(\cdot;\mu)$ (if it exists and is unique) to the probability measure $\mu$.  Then we denote it by ${\rm MLE}(\mu)\in D$.
 ${\rm MLE}(\mu)$ is {\it the circular Cauchy M-estimate for} $\mu$.
  For more informations on M-functionals, see, e.g., \cite{Huber}.

Following \cite{Douady1986,Chen2017} (see Section \ref{section::conformal-barycenter} and Definition \ref{defi::conformal-barycenter}), one can write ${\rm MLE}(\mu)=B(\mu)$, when $\mu$ has no atom. That is, the circular Cauchy M-estimate  of $\mu$ is the conformal barycenter $B(\mu)$ of the probability measure $\mu$. Computing the gradient of $L(\cdot;\mu)$ one obtains the related {\it circular Cauchy maximum likelihood equation}  which is given by
\begin{equation}\label{LikelihoodEquationGeneral}
\int_{\mathbb{S}^1}M_w(\xi){\rm d}\mu(\xi)=0.
\end{equation}
In the special case where $\mu$ is the empirical measure of the sample  
$\{\beta_1,\cdots,\beta_n\} $ associated with the base point $p$, this last equation becomes
\begin{equation}\label{LikelihoodEquationSample}
\sum_{j=1}^n M_w(\beta_j)=0.
\end{equation}

The following result will be useful in what follows:
\begin{lemma}\label{Uniform}
 Let $\mu$ be a law on $\S^1$ of unique circular Cauchy M-estimate  ${\rm MLE}(\mu)\in D$. 
Let $\beta$ be a random element taking values in $\S^1$ of law $\mu$. Then
\begin{equation}\label{Equi}
M_{{\rm MLE}(\mu)}(\beta) \hbox{ is uniform on }\S^1 \hbox{ if and only if }\mu \hbox{ is the circular Cauchy }C^*({\rm MLE}(\mu)).
\end{equation}
\end{lemma}
\begin{proof}
The right to left implication is trivial. Let us prove the converse one.
To see this, set $B={\rm MLE}(\mu)$ and assume $M_B(\beta)=U$, where $U$ is uniform on $\S^1$, or, equivalently the circular Cauchy $C^*(0)$. Then, using the above properties of the circular Cauchy family, 
$\beta = M_{-B}(U)$ is circular Cauchy $C^*(M_{-B}(0))=C^*(B)$.
\end{proof}

As we have seen in (\ref{RelationInit}), the initial conditions of the Kuramoto (o.d.e) (\ref{Kuramoto1}) are related to the base point $p$ as
$z_j(0)=\eta(0)M_{w(0)}(\beta_j)$, $j=1,\cdots, n$. Let 
$\mu_n^z$ be the empirical measure of the sample 
$\{z_1(0),\cdots,z_n(0)\}$. When they exist and are unique, we can relate the circular Cauchy M-estimates of the empirical measures $\mu_n$ and $\mu_n^z$ as follows.

\begin{lemma}
Assume ${\rm MLE}(\mu_n)$ exists and is unique. Then
\begin{equation}\label{RelationMLE}
{\rm MLE}(\mu_n^z)=\eta(0)M_{w(0)}({\rm MLE}(\mu_n)).
\end{equation}
\end{lemma}

\begin{proof}
For convenience set $w^*:={\rm MLE}(\mu_n)$.
We start from the likelihood equation (\ref{LikelihoodEquationSample}) and use (\ref{RelationInit}) to obtain that
$$0=\sum_{j=1}^n M_{w^*}(\beta_j)=\sum_{j=1}^n M_{w^*}(M_{-w(0)}(\overline{\eta(0)}z_j(0))).$$
A direct computation shows that
$$M_{w^*}(M_{-w(0)}(y))=\eta' M_{M_{w(0)}(w^*)}(y),$$
where $\eta'=(1-\overline{w(0)}w^*)/(1-w(0)\overline{w^*})$ is on the unit circle. Using the identity $M_a(e^{i\phi}z)=e^{i\phi}M_{ae^{-i\phi}}(z)$, one arrives at
$$0=\sum_{j=1}^{n}M_{\eta(0)M_{w(0)}(w^*)}(z_j),$$
which is the likelihood equation for the empirical measure $\mu_n^z$
of the circular Cauchy M-estimate $\eta(0)M_{w(0)}(w^*)$, and the result follows.
\end{proof}


\section{The central  angular Gaussian distribution\label{Section::centralAngularGaussian}}

Sections \ref{section::conformal-barycenter} and \ref{CircularCauchy} show that the limiting behavior of the solutions to the repulsive Kuramoto (o.d.e.) are well described using $M_{w(t)}(\beta_j)$, $j=1,\cdots,n$, where $w(t)$ converges toward the unique equilibrium state of (\ref{odew}) when the circular Cauchy M-estimate ${\rm MLE}(\mu_n)$ exists and is unique, where $\mu_n =\frac{1}{n}\sum_{j=1}^n \delta_{\beta_j}$. The next Sections show that such estimates are strongly related to central Gaussian  M-estimates for measures defined on the projective line $\RP$. We will in this way obtain new characterizations of phase-locked limiting states using notions from nonparametric multivariate statistics.

\subsection{The central angular Gaussian law on the projective line $\RP$}

We follow \cite{Auderset2005}. Let $X=(X_1,X_2)^T\in\R^2$ be a bivariate centred normal random vector of covariance matrix $\Sigma$, which is symmetric and positive definite. The {\it central angular Gaussian model} ${\cal G}_\Sigma$ is obtained from $X$ by retaining only its axis $[X]=\{\lambda X;\ \lambda\in\R\}$. The law of $[X]$ is called the \textit{univariate central angular Gaussian distribution of parameter $\Sigma$}. The sample space of the central angular Gaussian distribution is in fact the {\it projective space} $\RP=\{[x];\ x\in\R^2,\ x\ne 0\}$. A data point in $\RP$ can be seen as a pair of opposite unit vectors $\pm e^{i\theta}\in \S^1$.
As $[\lambda x]=[x]$ when $\lambda\in \R$, $\lambda\ne 0$, the covariance matrix $\Sigma$ of the centred normal vector $X$ is determined up to a positive constant. We thus assume that
$$\Sigma \in Pos^{+}(2):=\left\{ \begin{bmatrix}
a&c\\
c&b\\
\end{bmatrix} \; \vert\; a>0, b>0, \; ab-c^2=1\right\}.$$
The density of the unit vector $X/\vert\vert X\vert\vert$ with respect to the uniform probability measure on $\S^1$ is 
$(x^T \Sigma^{-1}x)^{-1}$, $\vert\vert x\vert\vert =1$. Thus, the density of the central angular Gaussian distribution ${\cal G}_\Sigma$ with respect to the uniform distribution on $\RP$ is
$$f_\Sigma^{\cal G}([x]):=\frac{x^T x}{x^T \Sigma^{-1}x},$$
$\Sigma\in Pos^{+}(2)$, $[x]\in \RP$.

\subsection{Relation to the Cauchy model on $\R$}

We arrive at the Cauchy model $C(\lambda)$, $\lambda =\mu + i\sigma$, that gives the law of the ratio $X_1/X_2$ by
setting (see, e.g. \cite{Auderset2005})
$$\Sigma := \frac{1}{\sigma}
\begin{bmatrix}
\sigma^2+\mu^2&\mu\\
\mu&1\\
\end{bmatrix}.$$

\subsection{The central angular Gaussian distribution on $\S^1$\label{CentralGaussianCircle}}

The authors of \cite{Kent1988} introduced similarly the notion of central angular Gaussian distribution on $\S^1$.
Let $y=e^{i \phi} \in \mathbb{S}^1$. The central angular Gaussian density of parameter $\Sigma =(\sigma_{ij})\in Pos^+(2)$ is  given by
\begin{equation}\label{AngularSphere}
h_{\Sigma}(\phi):=(2\pi)^{-1}  \frac{1}{\frac{1}{2}(\sigma_{11}+\sigma_{22}) +\frac{1}{2} (\sigma_{22}-\sigma_{11}) \cos(2\phi)- \sigma_{12} \sin(2\phi)}
\end{equation}
$0\le \phi< 2\pi$. Notice, the density $2 \pi h_{\Sigma}(\phi)$ is obtained from the previously defined central angular Gaussian distribution $f_\Sigma^{\cal G}([x])$  on $\RP$ by composing the density $f_\Sigma^{\cal G}([x])$ with  the projection ${\rm Proj}: \S^1 \to \RP$, with ${\rm Proj}(y)=[y]$; or otherwise saying the density of the central angular Gaussian  distribution  on $\RP$ equals $2 \pi h_{\Sigma}(\phi)$.

\subsection{The central angular Gaussian and the circular Cauchy models\label{CentralGaussianCircular}}

The authors of \cite{Kent1988} also make a link between the central angular Gaussian density of parameter $\Sigma=(\sigma_{ij}) \in Pos^{+}(2)$ evaluated at $y=e^{i\phi}\in\S^1$ and the circular Cauchy density of parameter $w=\rho e^{i\alpha}$, evaluated at $z=e^{i\theta}$, by setting
\begin{equation}\label{Square}
\theta = 2(\phi\ {\rm mod}(\pi)).
\end{equation}
Then
\begin{equation}\label{Link}
g_w(\theta)=h_\Sigma(\phi)
\end{equation}
where the respective parameters $w=\rho e^{i\alpha}$ and $\Sigma =(\sigma_{ij})$ are related as
\begin{equation}\label{Rel1}
2\frac{\rho\cos(\alpha)}{1+\rho^2}=\frac{\sigma_{11}-\sigma_{22}}{ \sigma_{11}+\sigma_{22}},
\end{equation}
and
\begin{equation}\label{Rel2}
2\frac{\rho\sin(\alpha)}{1+\rho^2}=\frac{2\sigma_{12}}{ \sigma_{11}+\sigma_{22}}.
\end{equation}

\begin{rem}\label{YSquare}
The form of the densities given in (\ref{CircularPolarDensity}) and
(\ref{AngularSphere}) show that if the unit random vector $Y=e^{i\Phi}$ follows a central angular Gaussian distribution on the circle $\S^1$ then $Z=Y^2 =e^{2i\Phi}= e^{i\Theta}$ is distributed according to a circular Cauchy law.
\end{rem}

\begin{rem}\label{Bijection}
Given the matrix $\Sigma=(\sigma_{ij})$ and solving the system given by equations (\ref{Rel1}) and (\ref{Rel2}), one finds the same $w=\rho e^{i\alpha}$  as given by the bijection (see (\ref{equ::pos_disc}) of the Appendix)

\begin{equation}
\label{equ::sigma_rho}
\sigma_{11}= \frac{1+\rho^2 +2\rho \cos(\alpha)}{1-\rho^2},\; \;
\sigma_{22}= \frac{1+\rho^2 -2\rho \cos(\alpha)}{1-\rho^2},\; \; 
\sigma_{12}=\frac{2\rho \sin(\alpha)}{1-\rho^2}.
\end{equation}
Conversely, we have:
\begin{equation}\label{equ:w}
w=\frac{1}{\sigma_{11}+\sigma_{22}+2}((\sigma_{11}-\sigma_{22})+2i\sigma_{12}).
\end{equation}

\end{rem}

\subsection{Central angular Gaussian M-functionals\label{MaxLikeliCentral}}

Let $\Prob$ be an arbitrary Borel probability measure on $\RP$. 
The {\it central angular Gaussian log-likelihood function} for $\Prob$
is the expectation of the logarithm of $f_\Sigma^{\cal G}$ evaluated at a random point $[x]\in \RP$ of law $\Prob$:
\begin{equation}\label{AngularLogLikelihood}
{\cal L}_{\Prob}^{\cal G}(\Sigma):=\int_{\RP} \ln(f_\Sigma^{\cal G}([x])){\rm d}\Prob([x]).
\end{equation}
The {\it central angular Gaussian M-functional} is the functional that assigns  the maximizer of ${\cal L}_{\Prob}^{\cal G}$ (if it exists and is unique) to any probability  $\Prob$ on $(\RP,\B(\RP))$.  We denote it by $\Sigma(\Prob)\in Pos^+(2)$, which is the {\it central angular Gaussian M-estimate for} $\Prob$.

A {\it projective subspace of dimension k} of the projective space $\RP$ is the set of axes $[x]\in \RP$ of the non-zero vectors $x$ lying in a linear space of dimension $k+1$ of $\R^2$, $k=0,1$. Theorem 1 of \cite{Auderset2005} provides a complete characterization of existence and uniqueness of  $\Sigma(\Prob)$  for multivariate central angular  Gaussian distributions. Here, we just recall the main relevant results in the univariate case:
\begin{itemize}
\item{} If the distribution $\Prob$ is absolutely continuous with respect to the uniform distribution on $\RP$, then $\Sigma(\Prob)$ is well-defined.
\item{} If the $\Prob$ probability of some point $[x]\in \RP$ is larger than $1/2$, then $\Sigma(\Prob)$ does not exist.
\item{} Let $\Prob_n = (\delta_{[x_1]}+\cdots +\delta_{[x_n]})/n$ be the empirical measure of a sample $\{[x_1],\cdots,[x_n]\}$ in $\RP$. Suppose that $n>2$ and that the sample is in general position, that is, that any non-trivial projective subspace of dimension $k$ of $\RP$ contains at most $k+1$ points of the sample. Then $\Sigma(\Prob_n)$ is well-defined.
\end{itemize}
Theorem 3 of \cite{Auderset2005} shows by computing the gradient of
${\cal L}_{\Prob}^{\cal G}(\cdot)$
 that the maximum likelihood equation, which is satisfied by $\Sigma(\Prob)$ is given by
\begin{equation}\label{TylerEquation0}
\frac{\Sigma}{2}= \E_\Prob \Big(\frac{Y Y^T}{Y^T \Sigma^{-1}Y}\Big) =
\int_{\RP}\frac{y y^T}{y^T \Sigma^{-1}y}{\rm d}\Prob([y]).
\end{equation}
 When $\Prob$ is the empirical measure of a sample $\{[y_1],\cdots,[y_n]\}$, $y_i\in\R^2$, (\ref{TylerEquation0}) becomes
\begin{equation}\label{TylerEquation}
\frac{\Sigma}{2}= \frac{1}{n}\sum_{j=1}^n \frac{y_j y_j^T}{y_j^T \Sigma^{-1}y_j},
\end{equation}
which is the so-called Tyler's equation (see \cite{Tyler1987A}). One should notice here that an estimator satisfying an equation like 
(\ref{TylerEquation}) is an {\it M-estimator} in the sense of Huber \cite{Huber}.

\begin{theorem}
\label{Gauss_Cauchy_max_like}
Let $f: \S^1 \to \S^1$ be the map $y \in \S^1 \mapsto f(y)=y^2$. Consider the projection
${\rm Proj}: \S^1 \to \RP$, with ${\rm Proj}(y)=[y]$. Let $\B(\S^1)$ and $\B(\RP)$ be the Borel $\sigma$-algebras associated with the metric spaces associated with $\S^1$ and $\RP$ for the euclidean metrics.
Suppose that the probability spaces 
$(\RP,\B(\RP),\Prob)$, $(\S^1,\B(\S^1),\nu)$ and $(\S^1,\B(\S^1),\mu)$ are related to each other as
$\Prob = \nu\ {\rm Proj}^{-1}$ and $\mu=\nu\ f^{-1}$. Then
$${\cal L}_{\Prob}^{\cal G}(\Sigma)=\int_{\RP} \ln(f_{\Sigma}^{\cal G}){\rm d}\Prob
=\int_{\S^1}\ln( f_w){\rm d}\mu+\ln(2\pi)=L(w;\mu)+\ln(2\pi).$$
where $w$ and $\Sigma$ are related as in (\ref{equ::sigma_rho}).
Then, if they exist and are unique, the M-estimate $\Sigma(\Prob) \in Pos^+(2)$ for the central angular Gaussian model equals, via the bijection (\ref{equ::sigma_rho}), the M-estimate ${\rm MLE}(\mu) \in D$ for the circular Cauchy model.
\end{theorem}

\begin{proof}  
Using the results of Sections \ref{CentralGaussianCircle}, \ref{CentralGaussianCircular},  \ref{MaxLikeliCentral}, and Remark \ref{Bijection}, one obtains that
\begin{eqnarray*}
{\cal L}_{\Prob}^{\cal G}(\Sigma)=\int_{\RP} \ln(f_{\Sigma}^{\cal G}){\rm d}\Prob &=& \int_{\RP} \ln(f_{\Sigma}^{\cal G}){\rm d}(\nu {\rm Proj}^{-1})=\int_{\S^1} \ln(f_{\Sigma}^{\cal G}\circ {\rm Proj}){\rm d}\nu\\
   &=&\int_{\S^1}\ln(2 \pi h_{\Sigma}){\rm d}\nu =\ln(2\pi)+\int_{\S^1}\ln( g_w \circ f){\rm d}\nu\\
   &=& \ln(2\pi)+\int_{\S^1}\ln( f_w){\rm d}(\nu f^{-1})\\
   &=&\ln(2 \pi)+\int_{\S^1}\ln( f_w){\rm d}\mu = \ln(2\pi)+L(w;\mu).
      \end{eqnarray*}
    \end{proof}

\subsection{Asymptotic properties of Kuramoto orbits with random base point\label{Section::AsymptoticMax}}

Coming back to the Kuramoto model, one can consider the  sample
$\{\beta_1,\cdots,\beta_n\}$ associated with the base point $p$,
 and the associated  empirical measure. Then by Theorem \ref{Gauss_Cauchy_max_like} the equilibrium point of (\ref{odew}), with base point given by $p$, corresponds to the M-estimate in the associated circular Cauchy model, or, to the barycenter of the empirical measure, or, to the M-estimate $\Sigma(\Prob_n)$ of the empirical measure of a sample $\{[y_1],\cdots,[y_n]\}$, $y_j=e^{i\phi_j}\in \S^1$ with $\beta_j =y_j^2$ from the central angular Gaussian model.

Assume that the circular M-estimate ${\rm MLE}(\mu_n)$ for $\mu_n$ exists and is unique. 
Let $\theta_j(t)$, $j=1,\cdots,n$ be the solution to (\ref{Kuramoto1})
in the repulsive case $\delta =-\pi/2$.
Let $z_j(t)=e^{i\theta_j(t)}$ with $z_j(t)=\eta(t) M_{w(t)}(\beta_j)$, where $\vert\eta(t)\vert =1$. Then, modulo rotations, the 
orbit $(z_1(t),\cdots,z_n(t))$ is phase-locked as $t\to\infty$
on the state described by the
 $M_{{\rm MLE}(\mu_n)}(\beta_j)$, $j=1,\cdots,n$. 

We  assume the $[y_j]$ are i.i.d. from a continuous distribution $\Prob$ on the projective space $\RP$, since any solution to (\ref{TylerEquation}) depends on the data points only through their directions. In what follows $\Prob_n$ denote the random empirical measure associated to the random sample. The following is Theorem 3.1 of \cite{Tyler1987B}:

\begin{theorem}\label{Consistency}
Let $\{[Y_1],\cdots,[Y_n]\}\subset \RP$ be an i.i.d. random sample, where the $Y_j$ are drawn from a continuous distribution $\Prob$ on $(\RP,\B(\RP))$. Then the random empirical measure $\Prob_n$ is such that
$\Sigma(\Prob_n)$ converges almost surely to $\Sigma(\Prob)$ as $n \to\infty$.
\end{theorem}
We can now describe the limiting properties of Kuramoto orbits for random base point $p=(\beta_1,\cdots,\beta_n)$. We assume in both results that the i.i.d. components $\beta_j$ of the base point are drawn at random according to a continuous probability measure $\mu$ on $(\S^1,\B(\S^1))$. The same results hold true for more general distributions, the main condition being that the related M-estimate exists and is unique (see Section \ref{MaxLikeliCentral}).

\begin{proposition}\label{KuramotoRandom1}
 Assume the  $\beta_j$ are i.i.d., distributed according to a continuous probability measure $\mu$ on $(\S^1,\B(\S^1))$, of barycenter ${\rm MLE}(\mu)$. Let $\mu_n$ be the empirical measure of the sample $\{\beta_1,\cdots,\beta_n\}$.  Then ${\rm MLE}(\mu_n)$  is the unique equilibrium of the (o.d.e.)
(\ref{equ::reduced}), and converges almost surely to ${\rm MLE}(\mu)$ as $n\to\infty$.  Moreover,
 $\sqrt{n}({\rm MLE}(\mu_n)-{\rm MLE}(\mu))$ is asymptotically  centred Gaussian.
 \end{proposition}

\begin{proof} We use the results from Sections \ref{CentralGaussianCircular} and \ref{MaxLikeliCentral}. 
Let $[y_1],\cdots,[y_n]$ be a sample of $\RP$ with $y_j=e^{i\phi_j}$, 
$\beta_j = e^{2i\phi_j}$  and let $\Prob_n$ and $\mu_n$ be the corresponding sample empirical measures.
Assume as in Theorem \ref{Gauss_Cauchy_max_like} that $Y$ is such that $\beta =f(Y)=Y^2$,
so that the laws $\nu$ of $Y$ on $\S^1$ and  $\Prob$ of $[Y]$ on $\RP$ are such that 
$\Prob =\nu {\rm Proj}^{-1}$ and $\mu=\nu f^{-1}$.
Theorem \ref{Gauss_Cauchy_max_like} shows that  ${\cal L}_{\Prob}^{\cal G}(\Sigma)=L(w;\mu)+\ln(2 \pi)$, where
the parameters $w$ and $\Sigma$ are related to each other through relation (\ref{equ::sigma_rho}). Similarly, 
 ${\cal L}_{\Prob_n}^{\cal G}(\Sigma)=L(w;\mu_n)+\ln(2 \pi)$.
The results of Section \ref{MaxLikeliCentral} show that $\Sigma(\Prob_n)$ (and therefore that ${\rm MLE}(\mu_n)$) exists and is unique when the sample $[y_1],\cdots,[y_n]$ is in general position, which holds almost surely true since by assumption the data are i.i.d. drawn from a continuous distribution.
Theorem \ref{Consistency} shows that $\Sigma(\Prob_n)$ converges almost surely toward $\Sigma(\Prob)$, which implies, using (\ref{equ:w}),  that  ${\rm MLE}(\mu_n)$ converges almost surely to ${\rm MLE}(\mu)$. The statistics 
$\Sigma(\Prob)$ is linearly invariant:  for any regular matrix $B$, let $\Prob^B$ denote the law of $[BX]$ when the random variable $[X]$ with values in $\RP$ is distributed accordingly to $\Prob$; then
$\Sigma(\Prob^B)=B \Sigma(\Prob) B^T$.  As $\Sigma(\Prob)\in Pos^{+}(2)$, one can find an orthonormal matrix $B$ such that $\Sigma(\Prob^B)={\rm id}_2$. 
Recall, the matrix exponential $\exp$ is a bijective map from the set of all symmetric $2 \times 2$ matrices to the set of all symmetric and positive definite $2 \times 2$  matrices; its inverse is thus denoted by $\ln$.
Theorem 6.11 and Remark 6.12 of \cite{Lutz2015} show that
$$\sqrt{n}(\ln(\Sigma(\Prob_n^B))-\ln({\rm id}_2) = 
\sqrt{n}\ln(\Sigma(\Prob_n^B)),$$
 is asymptotically centred Gaussian.
  Hence, setting
 $W_n = \sqrt{n}\ln(\Sigma(\Prob_n^B))$, one obtains that
$$\sqrt{n}(\Sigma(\Prob_n^B)-{\rm id}_2)=\sqrt{n} (e^{\frac{1}{\sqrt{n}}W_n}-{\rm id}_2)=W_n +{\rm O}(\frac{1}{\sqrt{n}}),$$
so that $\sqrt{n}(\Sigma(\Prob_n^B)-{\rm id}_2)$ is asymptotically centred Gaussian, and the same result holds true for 
$\sqrt{n}(\Sigma(\Prob_n)-\Sigma(\Prob))$. The corresponding result for 
the circular Cauchy model follows from the transformation given in (\ref{equ:w}), which is   smooth since ${\rm tr}(\Sigma)>2$ on $Pos^{+}(2)$.
\end{proof}

\begin{corollary}\label{Cor1}
Under the same hypotheses, the Kuramoto orbit $(z_1(t),\cdots,z_n(t))$ is asymptotically ($t\to\infty$, $n\to\infty$) phase-locked in the configuration 
$(M_{{\rm MLE}(\mu)}(\beta_1),\cdots,M_{{\rm MLE}(\mu)}(\beta_1)$. 
Moreover,
$M_{{\rm MLE}(\mu)}(\beta_j)$
is uniform on $\S^1$ if and only if $\mu$ is circular Cauchy.
 \end{corollary}
 
\begin{proof}
Consider (\ref{Mobius1}) and (\ref{Mobius}), which give that
$z_j(t)=\eta(t)M_{w(t)}(\beta_j)$, with $\vert \eta(t)\vert \equiv 1$,
and $w(t)\to {\rm MLE}(\mu_n)$, as $t\to\infty$ (the convergence follows from the fact that the (o.d.e.) (\ref{equ::reduced}) is of gradient form with a unique equlibrium).   Proposition \ref{KuramotoRandom1} gives the almost sure convergence of ${\rm MLE}(\mu_n)$ toward  the M-estimate ${\rm MLE}(\mu)$ as $n\to\infty$. The last statement is a consequence of Lemma \ref{Uniform}.
\end{proof}


\section{Phase-locked limiting states and spatial signs\label{Section::SpatialSigns}}

This Section characterizes the limiting phase-locked states of Kuramoto orbits. We will see that these limiting states are strongly related to so-called {\it spatial signs} $S_j\in\S^1$  which have been considered in nonparametric multivariate statistics.

\subsection{Tyler's transformation matrix and spatial signs}
\label{subsec::Tyler}

For a given sample $\{y_1,\cdots,y_n\}$, $y_j\in \R^2$, $j=1,\cdots,n$, let $\Sigma(\Prob_n) \in Pos^+(2)$  be the M-estimate of the central angular Gaussian model for the empirical measure $\Prob_n:= \frac{1}{n}\sum_{j=1}^n \delta_{[y_j]}$.  Consider any matrix $A_y$ such that
$$A_y^T A_y =\Sigma^{-1}(\Prob_n).$$
Such matrices $A_y$ are known as {\it Tyler transformation matrices} (and depend on the current sample). This transformation is used for many purposes in statistics, for example for the multivariate Sign test, see \cite{Randles2000,Oja2004}.

Consider the projection $S: \R^2\setminus\{0\} \to \mathbb{S}^1$ by $ x \mapsto S(x):=\frac{x}{\vert\vert x\vert\vert}$.
The spatial signs are defined as the affine invariant statistics
\begin{equation}\label{SpatialSigns}
S_j := S(A_y y_j),\ j=1,\cdots,n.
\end{equation}
Using the fact that $\Sigma(\Prob_n)$ satisfies (\ref{TylerEquation}), one obtains 
$$
\frac{\Sigma(\Prob_n)}{2}=\frac{A_y^{-1} (A_y^{-1})^T}{2}
= \frac{1}{n}\sum_{i=1}^n \frac{y_i y_i^T}{y_i^T A_y^T A_y y_i}
=\frac{1}{n}\sum_{i=1}^n \frac{y_i y_i^T}{\vert\vert A_y y_i\vert\vert^2},
$$
and thus
\begin{equation}\label{BasicSignEquation}
\frac{{\rm id}_2}{2}=\frac{1}{n}\sum_{i=1}^n S_i S_i^T.
\end{equation}
The transformed sample $S=\{S_1,\cdots,S_n\}$ is then such that
$\Sigma(\Prob_n^S)={\rm id}_2$, where $\Prob_n^S$ is the transformed sample empirical measure $\Prob_n^S: =\frac{1}{n}\sum_{j=1}^n \delta_{[S_j]}$.

\subsection{Phase-locked  states and spatial signs}

In this section we find an explicit relation between the spatial signs $\{S_1, \cdots, S_n\}$ and the base point $p=(\beta_1,\cdots, \beta_n)$ of the Kuramoto equation (\ref{Kuramoto1}). Let us consider the same notation and the hypotheses from Section \ref{subsec::Tyler} (see also Section \ref{sec::WA} from the Appendix).

\begin{theorem}\label{LinkB}
Let $y_j\in\R^2$, $j=1,\cdots,n$, and let $\beta_j \in \S^1$, $j=1,\cdots,n$,  be such that
$S(y_j)^2 = \beta_j$. Let $\mu_n$ be the empirical measure  of  the sample $\{\beta_1,\cdots,\beta_n\}$ and assume ${\rm MLE}(\mu_n) =:\rho e^{i\alpha}$ exists and is unique.
 Then, the spatial signs $S_j\in\S^1$ are such that
\begin{equation}\label{RelationS2M}
S_j^2 = e^{i(\pi-\alpha)} M_{{\rm MLE}(\mu_n)}(\beta_j),\ j=1,\cdots,n.
\end{equation}
 Let $z_j(t)=e^{i\theta_j(t)}$, $j=1,\cdots,n$, be solution of (\ref{Kuramoto1}) for the base point $p=(\beta_1,\cdots,\beta_n)$. Then, the Kuramoto orbit $(z_1(t),\cdots,z_n(t))$ is asymptotically phase-locked as $t\to\infty$ in the asymptotic configuration $(S_1^2,\cdots,S_n^2)$, $\forall n\ge 3$.
\end{theorem}

\begin{proof} Set for convenience $\frac{y_j}{\vert\vert y_j\vert\vert} = e^{i\phi_j}$ and $\beta_j =e^{2i\phi_j}$.
Then, by Section \ref{sec::WA}, the Tyler's transformation $A_y$ is related to  ${\rm MLE}(\mu_n)= \rho e^{i\alpha} $ as follows
$$ \pm A_y= \begin{bmatrix}
a&0\\
0&\frac{1}{a}\\
\end{bmatrix} \begin{bmatrix}
\cos((\pi-\alpha)/2+k\pi)&-\sin((\pi-\alpha)/2+k\pi)\\
\sin((\pi-\alpha)/2+k\pi)&\cos((\pi-\alpha)/2+k\pi)\\
\end{bmatrix}, $$
where $a=\sqrt{(1+\rho)/(1-\rho)}$.
This shows that, using complex notation, 
$\pm A_y \frac{y_j}{\vert\vert y_j\vert\vert}$ becomes
$$a \Re(e^{i (\frac{\pi-\alpha}{2} +\phi_j +k\pi)})+i \frac{1}{a} \Im(e^{i (\frac{\pi-\alpha}{2} +\phi_j +k\pi)}).$$
Let $\theta_{S_j}$ and $\theta_{S_j^2}$ be the arguments of the unit complex numbers $S_j$ and $S_j^2$. Then, from construction,
$$\tan\Big(\frac{\theta_{S_j^2}}{2}\Big)=\tan(\theta_{S_j})=
\frac{\frac{1}{a} \sin(\frac{\pi-\alpha}{2} +\phi_j +k\pi)}{a\cos(\frac{\pi-\alpha}{2} +\phi_j +k\pi)}=\frac{1}{a^2} \tan(\frac{\pi-\alpha}{2} +\phi_j),$$
or, equivalently,
\begin{equation}\label{IntermediateFormula}
\tan\Big(\frac{\theta_{S_j^2}}{2}\Big)=\frac{1-\rho}{1+\rho}\tan(\frac{\pi-\alpha}{2} +\phi_j).
\end{equation}

To go further, consider a relation of the form
\begin{equation}\label{WSM}
e^{i\theta}:=M_{-w}(e^{i\psi}e^{i\phi}), 
\end{equation}
where $M_{-w}$ is the inverse of $M_w$ in the sub-group $G$ of M\"obius transformations preserving the closure $\overline{D}$ of the Poincar\'e disc $D$ that are of the form (\ref{Mobius1}). Then
\begin{equation}
\label{equ::5}
e^{i(\theta-\alpha)}=\frac{\rho + e^{i(\phi -u)}}{1+\rho e^{i(\phi -u)}},
\end{equation}
where we set $u=\alpha-\psi$. The next step consists in using the trigonometrical identity
\begin{equation}\label{Trigo}
\tan(\frac{\theta-\alpha}{2})=i \frac{1-e^{i(\theta-\alpha)}}{1+e^{i(\theta-\alpha)}}. 
\end{equation}
Plugging the relation (\ref{equ::5}) into (\ref{Trigo}) one gets 
\begin{equation}\label{WSb}
\tan(\frac{\theta-\alpha}{2})=\frac{1-\rho}{1+\rho}\tan(\frac{\phi-\alpha+\psi}{2}).
\end{equation}
A direct comparison of (\ref{IntermediateFormula}) with
 (\ref{WSb}) shows that
$$S_j^2 = e^{-i\alpha} M_{-{\rm MLE}(\mu_n)}(-\beta_j)=e^{i(\pi-\alpha)} M_{{\rm MLE}(\mu_n)}(\beta_j),\ j=1,\cdots,n.$$
\end{proof}

\begin{corollary}\label{S2Uniform}
Assume the $\beta_j$ are i.i.d. drawn from a continuous distribution $\mu$ on $(\S^1,\B(\S^1))$. Then each $S_j^2$ converges almost surely as $n\to\infty$ to some limiting random variable $S^2(j)$ taking values in $\S^1$. $S^2(j)$ is uniform on $\S^1$ if and only if the $\beta_j$ are circular Cauchy.
\end{corollary}
\begin{proof} We know from Proposition \ref{KuramotoRandom1} that the random variables $M_{{\rm MLE}(\mu_n)}(\beta_j)$ converge toward $M_{{\rm MLE}(\mu)}(\beta_j)$ which are uniform if and only if the $\beta_j$ are circular Cauchy distributed. The result follows then from the identity (\ref{RelationS2M}).  
\end{proof}

Because $\sum_{j=1}^n M_{w^*(p)}(\beta_j)=\sum_{j=1}^n M_{{\rm MLE}(\mu_n)}(\beta_j)=0$, we have the following surprising result. Still, one can give an independent proof using Theorem \ref{Gauss_Cauchy_max_like}.
\begin{corollary}
\label{IncoherentSpatialSign}
We have that $\sum_{j=1}^n S_j^2=0$. In particular, 
 $S^2=(S_1^2,\cdots,S_n^2)$ belongs thus to the manifold
${\cal M}$ of incoherent states as given in (\ref{IncoherentManifold}).
\end{corollary}

\begin{proof}
By Theorem \ref{Gauss_Cauchy_max_like} and equation (\ref{BasicSignEquation}), so using the mapping $ Y \mapsto Y^2=Z$ which maps the central angular Gaussian model to the circular Cauchy model, and considering the related empirical measure $\mu_n(S_1^2,\cdots,S_n^2) = \frac{1}{n} \sum_{i=1}^n \delta_{S_i^2}$, we obtian  ${\rm MLE}(\mu_n(S_1^2,\cdots,S_n^2))=0$. So by (\ref{LikelihoodEquationSample})  we have $\sum_{j=1}^n S_j^2=0.$
\end{proof}

\section{Appendix} 
\label{appendix}

\subsection{Four models for the hyperbolic plane}

There are four models for the hyperbolic plane: 
\begin{enumerate}
\item
the Poincar\'{e} disc model $$D:= \{z \in \mathbb{C} \; \vert \; \vert z\vert <1\}$$ endowed with the Riemannian metric $ds^2= 4\frac{dx^{2}+ dy^2}{(1-x^2-y^2)^2}$, where $z=(x,y)$ are  the Cartesian coordinates of the ambient Euclidean plane;
\item
the Poincar\'{e} upper-half plane model $$\mathbb{H}=:\{(x,y) \;\vert \; x, y \in \mathbb{R}, y>0\} $$ endowed with the Riemannain metric $ds^2= \frac{dx^{2}+ dy^2}{y^2}$;
\item
the hyperboloid model $$\mathbb{H}^2=:\{(t,x,y) \in \mathbb{R}^3 \;\vert \; t^2-x^2-y^2=1, t>0\}$$ endowed with the Riemannian metric $ds^2= dt^2-dx^2-dy^2$;
\item
the positive definite, symmetric $2\times 2$ matrices model 
$$Pos^{+}(2):=\left\{ \begin{bmatrix}
a&c\\
c&b\\
\end{bmatrix} \; \vert\; a>0, b>0, \; ab-c^2=1\right\}$$ where at every point $\Sigma \in Pos^{+}(2)$ the Riemannian metric is given by $g_{\Sigma} (A,B)= trace(\Sigma^{-1}A \Sigma^{-1}B)$, for any vectors $A,B $ in the tangent plane $\mathrm{T}_{\Sigma}Pos^{+}(2)$ at $\Sigma$.
\end{enumerate}

Between these four models we have the following isometries:
\begin{enumerate}
\item
$h: D \to \mathbb{H}$, $z \mapsto h(z):= \frac{z+i}{iz+1}$, where $i \in \mathbb{C}$ is the imaginary unit;
\item
$g: D \to \mathbb{H}^2$, $z=(x,y) \mapsto g(z):=\left\{\frac{1+x^2+y^2}{1-x^2-y^2},\frac{2x}{1-x^2-y^2}, \frac{2y}{1-x^2-y^2} \right\}$. 

Its invers $g^{-1}: \mathbb{H}^2 \to D $ is $(t,x,y) \mapsto g^{-1}((t,x,y))= \left\{\frac{x}{1+t},\frac{y}{1+t} \right\}$;
\item
$f: \mathbb{H}^2 \to Pos^{+}(2)$, $(t,x,y) \mapsto f((t,x,y)):=\begin{bmatrix}
t+x&y\\
y&t-x\\
\end{bmatrix} $. 

Its invers $f^{-1}:  Pos^{+}(2) \to \mathbb{H}^2$ is given by
$ \begin{bmatrix}
a&c\\
c&b\\
\end{bmatrix} \mapsto f^{-1}( \begin{bmatrix}
a&c\\
c&b\\
\end{bmatrix})= (\frac{a+b}{2}, \frac{a-b}{2},c)$.
\end{enumerate}

In particular, we obtain $f\circ g :  D \to  Pos^{+}(2)$
\begin{equation}
\label{equ::pos_disc}
z \mapsto f(g(z))= {\huge \begin{bmatrix}
\frac{1+z \overline{z}+z+\overline{z}}{1-z \overline{z}}&\frac{z -\overline{z}}{i-iz \overline{z}}\\
\frac{z -\overline{z}}{i-iz \overline{z}}&\frac{1+z \overline{z}-z-\overline{z}}{1-z \overline{z}}\\
\end{bmatrix}}.
\end{equation}

\subsection{The relation between ${\rm MLE}(\mu_n)$ and Tyler's matrix $A_y$}
\label{sec::WA}

Let $y:=\{y_1,\cdots,y_n\}$, $y_j\in \R^2$, $j=1,\cdots,n$, be a sample with associated  M-estimate $\Sigma(\Prob_n) \in Pos^+(2)$ and empirical measure $\Prob_n:= \frac{1}{n}\sum_{j=1}^n \delta_{[y_j]}$ as in Section \ref{subsec::Tyler}. Denote $e^{i\phi_j}:=\frac{y_j}{\vert\vert y_j\vert\vert}\in\S^1$. Let 
$\beta_j: = e^{2i\phi_j}\in\S^1$ and $\mu_n:=\frac{1}{n}\sum_{j=1}^n \delta_{\beta_j}$ be the empirical measure of the sample $p:=\{\beta_1, \cdots, \beta_n\}$. 

By Theorem \ref{Gauss_Cauchy_max_like} the circular  M-estimate $B(\mu_n)=w^*(p) \in D$ of the measure $\mu_n$ equals $\Sigma(\Prob_n)$ in the hyperbolic plane model $Pos^+(2)$.  

To simplify the notation in the computation, in what follows consider $w=w^*(p) $ and $\Sigma= \Sigma(\Prob_n) $. Using the isometry $f\circ g :  D \to  Pos^{+}(2)$, given $w$ the positive definite symmetric matrix $\Sigma$ has the form
${\Large \begin{bmatrix}
 \frac{1+w \overline{w}+w+\overline{w}}{1-w \overline{w}}&\frac{w -\overline{w}}{i-iw \overline{w}}\\
\frac{w -\overline{w}}{i-iw \overline{w}}&\frac{1+w \overline{w}-w-\overline{w}}{1-w \overline{w}}\\
\end{bmatrix}} \in Pos^{+}(2).$

Similarly, to $-w$ corresponds $\Sigma^{-1}$ and 
 $$\Sigma^{-1}= {\Large \begin{bmatrix}
 \frac{1+w \overline{w}-w-\overline{w}}{1-w \overline{w}}&\frac{-w +\overline{w}}{i-iw \overline{w}}\\
\frac{-w +\overline{w}}{i-iw \overline{w}}&\frac{1+w \overline{w}+w+\overline{w}}{1-w \overline{w}}\\
\end{bmatrix}}.$$ 
As $\Sigma^{-1}= A_y^{T}A_y $, where $A_y=  \begin{bmatrix}
a&0\\
0&1/a\\
\end{bmatrix}\begin{bmatrix}
\cos(\theta)&-\sin(\theta)\\
\sin(\theta)&\cos(\theta)\\
\end{bmatrix} $ with $a \geq 1$, one has to solve the following equation to either find $A_y$ given $w$ or find $w$ given $\Sigma$:

\begin{equation}
\label{equ::M_A}
{\Large \begin{bmatrix}
 \frac{1+w \overline{w}-w-\overline{w}}{1-w \overline{w}}&\frac{-w +\overline{w}}{i-iw \overline{w}}\\
\frac{-w +\overline{w}}{i-iw \overline{w}}&\frac{1+w \overline{w}+w+\overline{w}}{1-w \overline{w}}\\
\end{bmatrix}}=  \begin{bmatrix}
\cos(\theta)&\sin(\theta)\\
-\sin(\theta)&\cos(\theta)\\
\end{bmatrix}  \begin{bmatrix}
a^2&0\\
0&1/a^2\\
\end{bmatrix} \begin{bmatrix}
\cos(\theta)&-\sin(\theta)\\
\sin(\theta)&\cos(\theta)\\
\end{bmatrix}. 
\end{equation}

By doing the computation for equation (\ref{equ::M_A}), given $A_y$ (i.e., $a>0$ and $\theta$) we obtain 
\begin{equation}\label{RelationWA}
w=\frac{a-\frac{1}{a}}{a+\frac{1}{a}} (\cos(-2\theta+ \pi)+i \sin(-2\theta+ \pi)).
\end{equation}

Now given $w= \rho e^{i\alpha} \in D, \rho<1$, we obtain $\pm A_y$: 
\begin{equation}
\label{equ::RelationAW}
a= \sqrt{\frac{1+\rho}{1-\rho}} \;  \text{ and } \theta \in \{(\pi-\alpha)/2+k\pi, \; \vert \; k \in \mathbb{Z}\}.
\end{equation}

Recall, when it exists and is unique,  $w^*(p) \in D$ is  the equilibrium for equation (\ref{incoherent1}): $$\sum_{j=1}^n M_{w^*(p)}(\beta_j)=0,$$
where $M_w(z):=\frac{z-w}{1-\bar wz}$, for $z \in \overline{D}$ and $w \in D$, is an isometry of the Poincar\'{e} disc model $D$. We obtain therefore a relation between $M_{w^*(p)}$ and $A_y$ via the relation between $w^*(p)$ and $A_y$.

 \bibliographystyle{abbrv} 
    

\end{document}